 \documentclass[final,3p,times]{elsarticle}

\usepackage{amssymb}

\usepackage[ruled,linesnumbered,noend]{algorithm2e} 
\usepackage{color}

\usepackage{verbatim}
\usepackage{amsmath}
\usepackage{amsmath,amsthm}
\newtheorem{theorem}{Theorem}
\usepackage{hyperref}
\usepackage[skins]{tcolorbox} 
\tcbuselibrary{breakable} 
\usepackage{graphicx}
\usepackage{subfigure}
\newtheorem{example}{Example}[section]

\begin{document}

\begin{frontmatter}



\title{TS-Cabinet: Hierarchical Storage for Cloud-Edge-End Time-series Database}

\cortext[cor1]{ Corresponding author at: Harbin Institute of Technology, No. 92, West Dazhi Street, Harbin, China.}
\author[label1]{Shuangshuang Cui}
\ead{cuishuangs@stu.hit.edu.cn}

\author{Hongzhi Wang\corref{cor1}\fnref{label1,label2}}
\ead{wangzh@hit.edu.cn}

\author[label1]{Xianglong Liu}
\ead{1190200525@stu.hit.edu.cn}

\author[label1]{Zeyu Tian}
\ead{2021110003@stu.hit.edu.cn}

\author[label1]{Xiaoou Ding}
\ead{dingxiaoou_hit@163.com}

\affiliation[label1]{organization={Harbin Institute of Technology},
	addressline={No. 92, West Dazhi Street},
	city={Harbin},
	postcode={150000},
	country={China}}

\affiliation[label2]{organization={Peng Cheng Laboratory (PCL)},
	addressline={No. 2, Xingke 1st Street, Nanshan},
	city={Shenzhen},
	country={China}}

\begin{abstract}
Hierarchical data storage is crucial for cloud-edge-end time-series database. Efficient hierarchical storage will directly reduce the storage space of local databases at each side and improve the access hit rate of data. However, no effective hierarchical data management strategy for cloud-edge-end time-series database has been proposed. To solve this problem, this paper proposes TS-Cabinet, a hierarchical storage scheduler for cloud-edge-end time-series database based on workload forecasting. To the best of our knowledge, it is the first work for hierarchical storage of cloud-edge-end time-series database. By building a temperature model, we calculate the current temperature for the time-series data, and use the workload forecasting model to predict the data's future temperature. Finally, we perform hierarchical storage according to the data migration policy. We validate it on a public dataset, and the experimental results show that our method can achieve about 94\% hit rate for data access on the cloud side and edge side, which is 12\% better than the existing methods. TS-Cabinet can help cloud-edge-end time-series database avoid the storage overhead caused by storing the full amount of data at all three sides, and greatly reduce the data transfer overhead between each side when collaborative query processing.
\end{abstract}




\end{frontmatter}


\section{Introduction}
At present, many industries are in urgent need of cloud-edge-end time-series database, such as intelligent transportation, intelligent manufacturing, etc. In a cloud-edge-end time-series database, the data are stored in the local database of the cloud computing center, edge devices, and end devices respectively. The end devices are responsible for data collection and generation, and then passed to the edge devices. The edge devices perform the preliminary calculations on the data and upload it to the data center. The data in cloud-edge-end scenarios come from diverse and large data sources. They contain real-time data acquired by a variety of sensors, mostly time-series data. It is a series of numerical data points collected over time. Each data point is collected at discrete time intervals. When data generation, each data point has a timestamp. 

The time series are typically stored at the cloud-edge-end according to the following strategy. The edge and end devices store the collected data or historical data for a long period time and periodically upload to the cloud side by the end devices. \textbf{If the data is accessed frequently, it means its temperature is high.} We call it hot data. On the contrary, we call it cold data. If an end device uploads all the data to the cloud and the edge devices, not only it causes a great waste of storage space, but also the query processing efficiency may be decreased by the frequent communication among the cloud-edge-end. Therefore, to reduce the storage space occupation and improve the query efficiency, the cloud-edge-end time-series database urgently needs an efficient hierarchical storage strategy to distribute the data among cloud-edge-end properly.

Hierarchical storage of time-series data in the cloud-edge-end time-series database brings three challenges as follows. (1) Firstly, how to determine the data to store at each side according to the features of devices, data, and workload. (2) The data access interval affects the data warming amplitude. Compared with traditional relational data, the query of time-series data usually appears periodically. To accurately calculate the data temperature, it is necessary to quantitatively represent the relationship between the access interval and the data temperature change in the temperature model. (3) Compared with traditional relational data queries, most temporal data queries to access data in batches by periods, and the queries are closely linked with timestamp attributes. To efficiently perform hierarchical storage scheduling of time-series data, it is necessary to predict the data temperature's rise and fall in the future.

Hierarchical data storage strategies have been well studied, but unsuitable for cloud-edge-end. To classify the hot and cold data, the cache replacement strategies represented by LRU and LFU use timeliness (the most recently accessed data is hot data) and access frequency (the data with the highest frequency in historical data is hot data). Specifically, \cite{chong} consider the timeliness and frequency of data access to classify the hot and cold data. In \cite{hash} and \cite{bloom}, hash-based and Bloom filter-based approaches are used to build a framework for identifying hot and cold data. Siberia proposed by Microsoft predicts the likelihood of data being accessed in the future by exponential smoothing to maximize the main memory hit rate~\cite{Siberia}. However, these methods can only portray the relative temperature between data and cannot accurately reflect the hot and cold degree of data. 

To describe the temperature of the data quantitatively, \cite{xie} and \cite{xu} modeled the temperature of data based on Newton's law, but such a model fails to describe the relationship between data temperature and access interval. In terms of hierarchical storage strategies, \cite{xie} proposed a temperature-model-based cache replacement strategy for hot and cold data migration, while \cite{xu} considered a high and low water level method for saturation monitoring of hot databases that can dynamically adjust the high and low water level thresholds. Gorilla stores the temporal data collected from the database in the last 26 hours in the hot storage media~\cite{Gorilla}. These methods are superior to indirect methods of measuring data temperature, but cannot predict the temperature of the data in future. 

From the above discussions, existing hierarchical storage methods cannot be directly applied in the cloud-edge-end time series database, since the temperature not only model ignores the relationship between data temperature and access interval, but also cannot predict the temperature change of data in future.

To address the above problems and realize efficient hierarchical storage for data in the cloud-edge-end time-series database, this paper proposes TS-Cabinet achieve the goal of minimizing storage space and maximizing data query efficiency with temperature modeling and workload prediction. To portray the relationship between data temperature and access interval, we first model data temperature based on Newton's law of cooling and the law of thermal radiation. The data temperature model is the foundation of hierarchical storage, and we can calculate the temperature of the data after workload prediction. We establish a query arrival rate prediction model with the proposed frequent access timestamp mining algorithm, which can predict the data temperature in the future. Such that we can calculate the temperature reached by the data in the future. Based on the predicted data temperature, we propose a data migration strategy which can select the best location for storage. With the known frequency of access to the data in the future, TS-Cabinet can improve the data access hit rate and accelerate the query.

The contributions of this paper are as follows.

(1) We propose TS-Cabinet to solve the difficulty of storing data hierarchically in the cloud-edge-end time-series database and propose a data migration strategy that can store data in the optimal location in the cloud-edge-end database.

(2) To reflect the relationship between data temperature and access interval, we propose a data temperature model based on the combination of Newton's law of cooling and the law of thermal radiation, and propose a data temperature update algorithm.

(3) To predict the temperature change of time-series data in the future, we propose a workload prediction method based on query logs and frequent access timestamp mining algorithm.

(4) We tested the performance of TS-Cabinet on a large-scale real dataset and showed that our method improved the hot data access rate by 12\% compared to other hierarchical storage methods.

The structure of this paper is as follows: Section 2 introduces the related work; Section 3 detailed the overall framework design of TS-Cabinet; Section 4 introduces the division method of hot and cold data; Section 5 introduces the workload prediction techniques; Section 6 conducts experiments and analyzes the experimental results; the last part concludes and discusses the future work.

\section{Related Work}

We classify the related work into cloud-edge-end database, hot and cold data hierarchical storage, and workload prediction to compare TS-Cabinet.

\textbf{Cloud-edge-end Storage.} Currently, the storage strategies for cloud-edge-end data are mainly divided into two categories: the centralized ones and the distributed ones. ~\cite{PB,OD,X} store the data obtained from sensor nodes in a centralized database. In comparison, the distributed-based method \cite{Sensor} can achieve better workload balancing than the centralized manners, but it is difficult to guarantee the consistency of the metadata due to the occurrence of link and node failures. There are already many database products supporting cloud-edge-end systems, but the current storage methods are mainly based on single cloud-side, edge-side or end-side storage. For example, OpenTSDB \cite{AS} only supports cloud-side data storage. It is worth noting that although Apache IoTDB ~\cite{IOT}, TDengine \cite{tao} have the ability to deal with the storage requirements of cloud-side and end-side data at the same time, they cannot meet the demand of collaborative processing.

\textbf{Hierarchical Storage of Hot and Cold Data.} The first step towards hierarchical storage of hot-cold data is to find an effective way to distinguish the hot-cold data. LRU and LFU, the representation of the classical cache replacement strategies, identify hot and cold data from the perspectives of timeliness and data access frequency; to maximize the main memory hit rate. Leveraging Newton's law of cooling as a tool to quantitatively describe two behaviors: data temperature decay with time and access temperature rising, ~\cite{xie,xu} explicitly build a data temperature model. In terms of a tiered storage strategy, \cite{xie} proposes a temperature model-based cache replacement strategy to realize hot-cold data migration. Faced with the need for hot-cold data identification and hierarchical storage on a temporal database, Gorilla \cite{Gorilla} takes the last 26 hours of temporal data collected from monitor devices as hot data and saves in memory.

\textbf{Workload Forecasting.} For workload prediction tasks, \cite{CIKM} proposes to identify and forecast significant workload migrations depending on SQL-based workloads monitoring with Markov chains, neural networks, and fuzzy logic. QueryBot5000 \cite{QB} uses an online dynamic clustering approach to process SQL queries, and employs prediction models for each cluster. \cite{RNN} put forward an idea about using an RNN to predict future workloads of DaaS, while \cite{TEALED} leverage an online, multi-step workload prediction method based on the auto-encoder framework to forecast the resource utilization and query arrival rate of DBMS.

\section{Overview}
In this section, we overview TS-Cabinet as a hierarchical storage scheduler for cloud-edge-side time-series database. For efficient hierarchical storage of data in the cloud-edge-end time-series database, TS-Cabinet calculates the current temperature for time-series data and also predicate the future temperature for the data, and finally migrate the data according to both the current and future temperature. The architecture of the TS-Cabinet is shown in Figure 1, which consists of three main modules: hot-cold data classification, workload forecasting and data migration.

\begin{figure}[t]
	\centering
	\setlength{\abovecaptionskip}{0cm} 
	\setlength{\belowcaptionskip}{0cm}
	\includegraphics[width=0.9\textwidth]{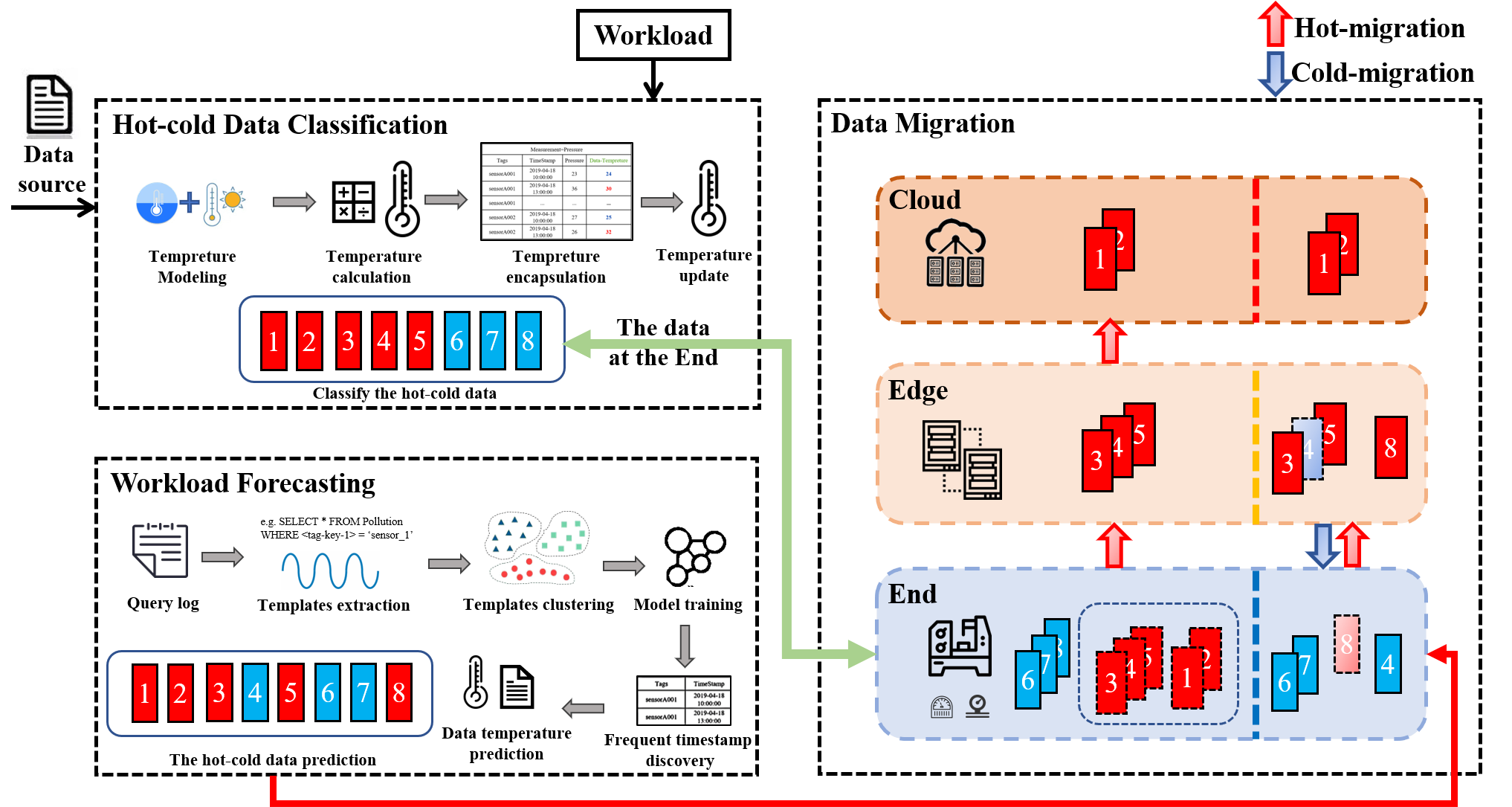}
	\caption{The Architecture of TS-Cabinet.} \label{fig1}
\end{figure}

To address the relationship between data temperature and access interval is difficult to be reflected in the temperature model, the first module of TS-Cabinet constructs a temperature model based on Newton's law of cooling and the law of thermal radiation to portray it to divide the data into hot and cold one. This module will be introduced in detail in Section 4. To predict the temperature change of time-series data in the future, workload forecasting module forecasts the data temperature in the next period time by building a query history arrival rate prediction model and a frequent access timestamp mining algorithm. This module will be described in detail in Section 5. 

Based on the hot-cold data classification and workload forecasting, the data migration module can select the best location for data to store. In this module, data migration mainly happens between cloud, edge, and end devices. Since computation is usually performed at the edge and the cloud devices, TS-Cabinet stores hot data in advance in a hierarchical manner at the edge and the cloud, which can reduce the size of the transferred data. TS-Cabinet stores the cold data at the end devices and creates a data summary  for over-cooled data to minimize storage space of the end. Workload prediction can help us to predict the temperature of data in the next time period in advance. If the temperature is too high, the data will be preheated to the edge and cloud devices in advance. If the temperature decreases, proving that it has a low probability of being accessed in the future, the data remains stored at the end device.

\section{Hot-cold Data Classification}

The hot and cold data classification is the base for hierarchical storage in the cloud-edge-end time-series database. In order to represent the relationship between data temperature and access interval, we model the temperature based on Newton's law of cooling and the law of thermal radiation. First, we bring the metrics of timeliness and frequency of data access for dividing hot and cold data into the model, i.e., data write time, frequency of data being accessed, etc. On this basis, we use timestamp to represent the data write time to reduce the storage overhead. Second, we represent the relationship between access interval and data temperature with the law of thermal radiation. Third, to reduce storage overhead further, the data temperature is encapsulated with the original data and stored together. In sum, the data hierarchical storage can facilitate access to the temperature of the data at the current moment.

Next, we describe the temperature model for time-series data in detail.  Newton's law of cooling describes an object, when its temperature is higher than the ambient temperature, will discharge Arnie's heat to the surrounding environment and gradually cool down, and the cooling rate is proportional to the temperature difference between the object and the ambient temperature. So we construct the temperature model based on Newton's law of cooling for the temperature decay part. On this basis, we introduce a ``heat source" to portray the heating process brought about by accessing data, which is equivalent to adding a heat source to the data. When the data is written, its temperature is computed according to Theorem 1 if it is accessed. 
\begin{theorem}
	The $data_i$ accessed $s$ times in $(t_n - t_{n-1})$ , the temperature of $data_i$ is $T_i(t_n)$.
	\begin{equation}
	T_i(t_n) = T(t_{n-1})e^{-k(t_n - t_{n-1})} + \gamma \times {T_{heat}}^4 \times s /  (t_n - t_{n-1})
	\end{equation}
	where $T_{heat}$ is temperature of the heat source, $k$ denotes the cooling rate and $\gamma$ denotes the heating rate.
\end{theorem}

\begin{proof}
	Eq.(2) give Newton's law of cooling.
	\begin{equation}
	d T(t)/d t =-a (T(t) - H)
	\end{equation}
	where $H$ is the ambient temperature, $T(t)$ is the current temperature of the object, and $a$ is the ratio coefficient of the temperature change rate to the temperature difference. By solving the differential equation, we obtain Eq.(3).
	\begin{equation}
	T(t) = (T_0 - H)e^{-kt} + H
	\end{equation}
\indent We use Eq.(3) to model the temperature decay term $T_f(t_n)$ and neglect the effect of ambient temperature. we obtain Eq.(4), $k$ denotes the cooling rate.
\begin{equation}
T_f(t_n) = T(t_{n-1})e^{-k(t_n - t_{n-1})}
\end{equation}

\indent In the thermal radiation law, The heat radiation per unit area is Eq.(5).
\begin{equation}
f(x,T_{heat}) = \sigma \times {T_{heat}}^4r^2/x^2
\end{equation}
Where $T_{heat}$ is the temperature of the heat source, $r$ is the radius of heat source, $x$ is the distance from the heat source, and $\sigma$ is the thermal radiation coefficient, which is a constant in thermodynamics. The heat radiation transfers heat, which is in Eq.(6).  
\begin{equation}
f(x,T_{heat})\times ds\times dt = dQ = C\times dm\times dT_r
\end{equation}
The actual temperature raised $T_r$ by the data is Eq.(7).
\begin{equation}
T_r = f(x,T_{heat})(t_n - t_{n-1})/Cd\rho
\end{equation}
where $C$ is Specific Heat Capacity, which is the heat capacity per unit mass of a material. $d$ is the thickness of the object. $\rho$ is Density, which is mass per unit volume. For calculating data temperature, $C$,$d$ and $\rho$ are constant.

\indent Let $\beta = \sigma r^2/ C d$, and $(t_n - t_{n-1})$ is modeled as the distance $x$ between the data and the heat source.  $T_r$ is Eq.(8)
\begin{equation}
T_r = \beta \times {T_{heat}}^4 / \rho x
\end{equation}
The data temperature model for TS-Cabinet includes both temperature decay and warming. $T(t_n)$ is Eq.(9) and (10), $c$ is a discrete function in Eq.(11)
\begin{equation}
T(t_n) = T_f + T_r \times c
\end{equation} 
\begin{equation}
T(t_n) = T(t_{n-1})e^{-kx} + \beta \times {T_{heat}}^4 \times c / \rho x
\end{equation}

\begin{eqnarray}
c= \begin{cases}
1,\quad  \mathrm{When ~ data ~is~ accessed~ at~ t_n} \\
0,\quad  \mathrm{When ~ data ~is ~not ~accessed ~at~t_n}
\end{cases}
\end{eqnarray}

The $\rho$ of data is considered as a constant here. Let $\gamma = \beta / \rho $, $T(t_n)$ is in Eq.(11)
\begin{equation}
T(t_n) = T(t_{n-1})e^{-kx} + \gamma \times {T_{heat}}^4 \times c /  x
\end{equation}
\begin{equation}
T_i(t_n) = T(t_{n-1})e^{-k(t_n - t_{n-1})} + \gamma \times {T_{heat}}^4 \times s /  (t_n - t_{n-1})
\end{equation}
Eq.(12) represents the change of temperature in $(t_n - t_{n-1})$ when the heat source heats the $data_i$ one time, that is, the $data_i$ is accessed one time. If accessed $s$ times in $(t_n - t_{n-1})$, the temperature of $data_i$ is $T_i(t_n)$.\qed
\end{proof}

We have solved the difficulty when the $T_r$ in Eq.(9) of the existing method is constant, which is only related to the number of access instead of the access interval. We add a time interval to limit the $T_r$, so that the more data accessed in the same period time, the larger their $T_r$ is. Conversely, the fewer data accessed in the same period time, the smaller their $T_r$ is.

For the temperature model, it is difficult to decide how often to update the temperature. If we update a large amount of time-series data temperature all the time, it will be a huge overhead on computing resources. Since the time-series data query is periodic, we can update temperature periodically. Thus we update the data temperature within a time window. The updated data temperature algorithm is shown in Algorithm 1.\\
\begin{algorithm}[t]
	\SetAlgoLined 
	\caption{Data temperature updating algorithm}
	\KwIn{Temperature of the heat source: $T_{heat}$, Time window: $w=t_n-t_{n-1}$, Times of access: $s$}
	\KwOut{Data temperature at $t_n$: $T_i(t_n)$}
	insert data $x_i$\;
	$T_0$ = $T_{heat}$\; 
	$s$ = 0\;
	Maintain~a~tuple~($Timestamp_{query}(t_{n-1})$, $T_i(t_n)$ )~for each time series\;
	\For{i=0;\,i\,\textless \,data.num\,;\,i++}{
		\If{Time interval\textless $w$}{
			\If{ query the Data}{
				$s = s+1$\;
			}
		}
		
		\ElseIf{Time interval= $w$ }{
			$T_i(t_n) = T(t_{n-1})e^{-k(t_n - t_{n-1})} + \gamma {T_{heat}}^4 *s /  (t_n - t_{n-1})$\;
			$s = 0$\;
			update the tuple
			($Timestamp_{query}(t_{n-1})$, $T_i(t_n)$ )\;
		}
	}
	return
\end{algorithm}
When the data $x_i$ is inserted, its initial temperature is set to be the same as the heat source temperature. And we maintain a tuple for $x_i$ to record the timestamp of the last access and the current temperature(Line 4). When the query for $x_i$ arrives, we do not update the temperature of $x_i$ immediately, but make the $s$ to $s$+1(Lines 6-8). We update the temperature records of all data uniformly when the time of our specified time window arrives, and return $s$ to 0(Lines 9-12). The time complexity of Algorithm 1 is $O(n)$.
\begin{table}[htbp]
	\centering
	\setlength{\abovecaptionskip}{0cm} 
	\setlength{\belowcaptionskip}{0cm}
	\caption{Temperature storage structure}\label{tab1}
	\begin{tabular}{|l|l|l|l|l|l|}
		\hline
		\multicolumn{6}{|c|}{Measurement}\\
		\hline
		Tags & Timestamp & Field1 & ... &  \textbf{Last-query Timestamp} & \textbf{Data-Temperature}\\
		\hline
	\end{tabular}
\end{table}

In the implementation, to facilitate the temperature update and to obtain the data temperature in real-time, TS-Cabinet encapsulates the data temperature and adds two new fields after the temporal data fields, namely Last-Query Timestamp and Data-Temperature in Table1. They are used to store the current timestamp and the current temperature of the data being accessed. Query Timestamp type is \textsf{TIMESTAMP}, which occupies 4 bytes of storage space, Data-Temperature type is \textsf{FLOAT}, which occupies 4 bytes, so the increased storage overhead is the amount of temporal data $\times$ 8bytes. In order to reduce the storage overhead, the Last-Query Timestamp and Data-Temperature storage of cold data will not be maintained when the data temperature is too cold.
\section{Workload Forecasting}
Most of the time-series data are accessed in bulk by period time. If we can predict the probability of data being accessed in the future time and catch the temperature rise and fall of data in future periods, we can improve the efficiency of hierarchical storage scheduling of time-series database. To know how the data temperature will change in the future, we develop query arrival rate prediction and frequent timestamp mining algorithm to forecast the workload in TS-Cabinet. In this section, we first introduce query generation to solve the problems in the lack of query historical workload in time-series database (Section 5.1). In order to reduce the model prediction time, we introduce how to extract the query template and take it into cluster after generating the query (Section 5.2). Next, to obtain frequently accessed data, we introduce the forecasting models for the query template cluster (Section 5.3) and the estimation of the frequent access timestamp (Section 5.4).

\subsection{Query Generation}
Since the query logs are difficult to obtain and handle in practice, query generation is a method to replace query logs. However, the current mainstream query generation~\cite{R1,R2} is mainly for relational data, which is difficult to be directly applied to the time series. Therefore, researchers try to use synthetic query sets for benchmarking on time-series databases. Time series database Benchmark (like YCSB-TS~\cite{YCSB}, TSBS~\cite{TSBS}, and TS-Benchmark~\cite{TSB}) are all based on pre-set query templates to generate queries. However, the generated queries can only cover a small share of simple cases and cannot meet the needs of our workload prediction requirements. To address these problems, we propose a template-based workload generation method for time series. We attempt to discover the patterns of query arrival rates in existing benchmarks such as Cycles, Stability, Spike, and Chaos, and construct the templates with such arrival rate patterns.

\textbf{Query Template Construction.}  To construct query template, it is essential to consider enough cases of query workload comprehensively. To generate comprehensive and diverse templates, we divide the time-series data queries into three categories: Conditional Query, Aggregate Query, and Group Query. Based on these three kinds of queries, we design five query templates. We will introduce them with the pollution dataset~\cite{pol} whose attributes are in Table 2 as examples.
\begin{table}
	\centering
	\setlength{\abovecaptionskip}{0cm} 
	\setlength{\belowcaptionskip}{0.2cm}
	\caption{Pollution attributes transform to influxdb's format}\label{tab1}
	\begin{tabular}{|l|l|l|l|l|l|l|l|}
		\hline
		Ozone & particulate matter & $CO_2$ & $SO_2$ & $NO_2$
		& longitude & latitude & timestamp\\
		\hline
		\multicolumn{5}{|c|}{Field} & \multicolumn{2}{|c|}{Tag}  & Timestamp\\
		\hline
	\end{tabular}
\end{table}
\\\indent\textbf {Query Template 1.} Conditional queries on single time-series data.

\begin{example}
	When people want to travel, they execute this query to find out the air quality changes at the destination over a period of time. Experts predict the future air quality of a place by analyzing the change of air conditions in the past.\\
	\fbox{%
		\parbox{1\textwidth}{%
			\textbf{SELECT }\textless \textup{field-key}$\colon$ \textup{ozone}\textgreater $\lbrack$ ,\textless $field-key$\textgreater,\textless $tag-key$\textgreater,...\textup{or}*$\rbrack$ \textbf{ FROM}\,\\\rm pollution 
			\textbf{ WHERE } \textless $tag-key$$\colon$ $sensor-id$\textgreater= \textless $tag-value$$\colon$ $id$\textgreater
			\textbf{ AND } \textup{time} $\geq$ ? 
			\textbf{ AND } \textup{time} $\leq$ ? 
		}%
	}
\end{example}

\textbf {Query Template 2.} Conditional queries on multiple time-series data by time period. 
\begin{example}
If there are multiple sensors at the place, the query will visit the corresponding data of multiple sensors when experts plan to observe the change in air quality in multiple places.\\
	\fbox{%
		\parbox{1\textwidth}{%
			\textbf{SELECT }\textless  \textup{field-key}$\colon$ \textup{ozone}\textgreater  $\lbrack$,\textless $field-key$\textgreater,\textless $tag-key$\textgreater,...\textup{or}* $\rbrack$\textbf{ FROM } \,\\\rm pollution
			\textbf{ WHERE } \textless$tag-key$ $\colon$ $sensor-id$\textgreater=\textless$tag-value$$\colon$$id$\textgreater \textbf{ IN }$\lbrace$ \textless $tag-id1$\textgreater,… $\rbrace$
			\textbf{ AND }\textup{time} $\geq$ ? 
			\textbf{ AND }\textup{time} $\leq$ ? 
		}%
	}
	
\end{example}
\textbf {Query Template 3.} Conditional queries for time-series data by period time and certain operators. 
\begin{example}
	When an air pollutant indicator exceeds or falls below a threshold value, the query needs to be executed to determine the specific time period and to conduct further analysis and warning. Operator can be \textgreater, $\geq, \textless, \leq and =.$\\
	\fbox{%
		\centering
		\parbox{1\textwidth}{%
			
			\textbf{SELECT }\textless  \textup{field-key}$\colon$ \textup{ozone}\textgreater $\lbrack$,\textless $field-key$\textgreater,\textless $tag-key$\textgreater,...\textup{or}* $\rbrack$\textbf{ FROM } \,\rm \\pollution
			\textbf{ WHERE } \textless $tag-key$ $\colon$ $sensor-id$\textgreater= \textless $tag-value$ $\colon$ $id$\textgreater
			\textbf{ AND }\textup{time} $\geq$ ? 
			\textbf{ AND }\textup{time} $\leq$ ? 
			\textbf{ AND } \textless $field-key$$\colon$ ? \textgreater \textbf{ OPER} ? 
			
		}%
	}
	
\end{example}

\textbf {Query Template 4.} Aggregate queries for time-series data by period time.\
\begin{example}
 People can know the approximate air pollution situation of a place over a period time where the Operator can be AVG, MAX, MIN.\\
	\fbox{%
		\centering
		\parbox{1\textwidth}{%
			
			\textbf{SELECT OPER? }(\textless  \textup{field-key}$\colon$ \textup{ozone}\textgreater  $\lbrack$,\textless $field-key$\textgreater,\textless $tag-key$\textgreater,...\textup{or}* $\rbrack$)\\\textbf{FROM } \,\rm pollution
			\textbf{ WHERE } \textless $tag-key$ $\colon$ $sensor-id$\textgreater= \textless $tag-value$ $\colon$ $id$\textgreater
			\textbf{ AND } \\\textup{time} $\geq$ ? 
			\textbf{ AND } \textup{time} $\leq$ ? 
		}%
	}
\end{example}

\textbf {Query Template 5.} Group queries for time-series data by period time.
\begin{example}
	People compare the air quality of different places to decide their final destination before traveling.\\
	\fbox{%
		\centering
		\parbox{1\textwidth}{%
			
			\textbf{SELECT OPER?}(\textless \textup{field-key}$\colon$ \textup{ozone}\textgreater$\lbrack$,\textless $field-key$\textgreater,\textless $tag-key$\textgreater,...\textup{or}* $\rbrack$)\\\textbf{FROM } \,\rm pollution
			\textbf{ WHERE } \textless $tag-key$ $\colon$ $sensor-id$\textgreater= \textless $tag-value$ $\colon$ $id$\textgreater \textbf{ IN } \\$\lbrace$ \textless $tag-id1$\textgreater,\textless $tag-id2$\textgreater,… $\rbrace$
			\textbf{ AND } \textup{time} $\geq$ ? 
			\textbf{ GroupBy } \textup{Clause}
		}%
	}
\end{example}

\textbf{Arrival Pattern.} To simulate real-world scenarios, we describe four common workload patterns that are popular temporal database applications.

\textbf{Cycles:} The number of queries varies periodically. For example, during the day, 7-9 o'clock and 18-21 o'clock may be the peak time of the query. For the same query template, there may be multiple different cycles.

\textbf{Stability:} The number of queries is stable and doesn't change over time. When experts want to analyze sensor data, they usually execute queries at specific moments (i.e. every half hour). The arrival rate of queries is stabilized.

\textbf{Spike:} The number of queries shows explosive growth and decay. When the sensor is abnormal, for timely detection and repair of abnormalities may lead to a large number of queries in a short period time.

\textbf{Chaos:} Queries are generated randomly and the arrival rate is irregular.

When query templates and arrival patterns are available, query generation can be performed. Query generation is divided into four main steps: first of all, constructing a summary of query templates for the dataset, then matching the arrival rate pattern for each template, determining the number of queries generated by a template over a period of time based on the arrival rate pattern, and finally populating the templates to generate a collection of queries. The time complexity of Algorithm 2 is $O(n)$. Query generation can  build query templates automatically, and  determine the corresponding query arrival rate pattern for each template automatically. It can also generate high-quality and diversified queries for each template based on the reach pattern.\\
\begin{algorithm}[H]
	\SetAlgoLined 
	\caption{Query Generation}
	\KwIn{Dataset, Arrival pattern}
	\KwOut{query set}
	Build Template Summary for the dataset\; 
	\For{i=0;\,i\,\textless \,template.num\,;\,i++}{
		match (template.i,random(Arrival Pattern))\;
	}
	return
\end{algorithm}
\subsection{Query Template Extraction and Clustering}
Since the large number of queries will consume too much time in training models, we designed the query template extraction method and the query template clustering method to accelerate the prediction.

\textbf{Query Template Extraction.} We extract the template for each input query with the following two steps. First, we record the timestamp of the query, $record=(Now()-time_{end}, Now()-time_{start})$. Then we replace the timestamp information of the query with a placeholder. Since we only need to predict the frequently accessed data by the workload and do not care further operations on the data such as aggregation, we remove such operations. An example of query template extraction is shown in Figure 3, where the timestamp is replaced with a placeholder, and the AVG operation in the query is removed.

\begin{figure}
	\centering
	\setlength{\abovecaptionskip}{0cm} 
	\setlength{\belowcaptionskip}{0cm}
	\includegraphics[width=0.9\textwidth]{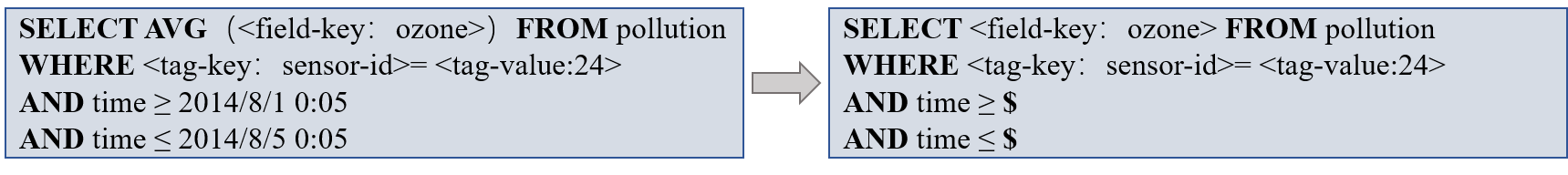}
	\caption{Query Template Extraction.} \label{fig2}
\end{figure}

\textbf{Query Template Clustering.} TS-Cabinet uses a QB5000-based clustering method~\cite{QB} for clustering query templates. The original DBSCAN algorithm uses the minimum distance between an object and any object in a cluster. However, it has uncertainty. To avoid randomness in the clustering process, we specify the distance between object and cluster center as the evaluation criterion, which represents the template of cluster. The DTW distance is used as a similarity measure between the historical arrival rates of two templates. Since it can solve complex time series prediction problems~\cite{DTW}. Only if the similarity of the query templates within each cluster is high, the prediction arrival rate can be more accurate when deploying a prediction model for them. For this purpose, we add a threshold $\rho$ to the similarity metric, and only when the DTW distance between the new query templates and center of the cluster gets less than $\rho$, the new query templates will belong to this cluster.

Query Template Clustering Algorithm is shown in Algorithm 3. For each new template, the distance between its historical arrival rate and the center of any cluster is first calculated, and the template is assigned to the cluster with the smallest distance and less than $\rho$. We use a kd-tree to quickly find the closest center of existing clusters to the template in a high-dimensional space~\cite{KD}. If there are no existing clusters (this is the first query), or if no cluster has a center close enough to the template, TS-Cabinet will create a new cluster with the template as the only member(Line 1-10). TS-Cabinet periodically checks the similarity of the previous templates to the centers of the clusters that they belong to. If the similarity of a template is no longer less than $\rho$, TS-Cabinet removes it from the current cluster and repeats Lines 1-10 to find a new cluster(Lines 11-13). If the cluster does not receive a  template for a long period time, TS-Cabinet will delete the cluster(Lines 14-15). The complexity of these steps is $O(nlogn)$, where $n$ is the number of templates in the workload.
\begin{algorithm}[t]
	\SetAlgoLined 
	\caption{Query Template Clustering Algorithm}
	\KwIn{$templates$ is the new templates set, threshold $\rho$ }
	\KwOut{$clusters = \left\{c_1,c_2,...,c_m\right\}$ is the cluster set}
	$templates$ marked as unprocessed \; 
	\For{i=0;~i~\textless \,$templates.num$\,;\,i++}{
		\If{$template_i$ \,is \,the\, first\, query}{
			Create $c_m$ for $template_i$\;
		}
		\Else{
			Calculate the $DTW$ between $template_i$ and each center of $clusters$\;
			
			\If{$DTW_i$ between $template_i$ and $center_m$ of $clusters$ is minimum and $DTW_i$\textless $\rho$ }{
				$c_m$ $\gets$ $template_i$\;
			}
			\Else{
				Create $c_m$ for $template_i$\;
			}
		}
		\If{the $DTW$ between $template_i$ and $center_m$ $\geq$ $\rho$}{
			Drop $template_i$ from $c_m$ \;
			Repeat Line3-14 for the $template_i$ \;
		}	
		\If{ $C_m$ does not receive a template for a long time}{
			Drop the $c_m$ from $clusters$\;
		}	
	}
	return
\end{algorithm}
\subsection{Forecasting Models}
To design the prediction model for query arrival rate, we analyze and compare the four most widely used prediction models including Linear, ARIMA, Holt winter's and LSTM. We found the Linear and LSTM outperformed the other models. First, Linear can avoid overfitting well when dealing with short-term prediction because of the intrinsic relationship of data is simple, so Linear short-term prediction is very effective. As for long-term prediction, Long Short Term Memory (LSTM) as a variant of RNN enables the network to automatically learn the periodicity and repetitive trends of data points in time series, which is more suitable for predicting nonlinear scenarios compared to traditional RNN. \\
\indent Then how to combine the advantages of multiple models? In other application fields, the effective method of researchers is to use ensemble model, which combines multiple models for average prediction. In prediction tasks, ensemble methods combine multiple machine learning techniques into a single prediction model to reduce variance or bias (e.g., Boosting)\cite{DW}. Previous work has shown ensemble methods work well and they are often the top winners in data science competitions~\cite{RP,ZH}. To combine the advantages of LR and LSTM, TS-Cabinet uses an ensemble learning model with a combined Linear and LSTM. The Linear model is shown in Eq.(14), where $a, b_1, b_2, ... , b_n$  is the parameter, called the regression parameter. The LSTM model is shown in Eq.(15-20). 
\begin{equation}
y = a + b_1x_1 + b_2x_2,...,b_n x_n + \varepsilon
\end{equation}
\begin{equation}
f_t = \sigma (W_{fx}x_t + W_{fh}h_{t-1} + b_f)
\end{equation}
\begin{equation}
i_t = \sigma (W_{ix}x_t + W_{ih}h_{t-1} + b_i)
\end{equation}
\begin{equation}
g_t = tanh (W_{gx}x_t + W_{gh}h_{t-1} + b_g)
\end{equation}
\begin{equation}
O_t = \sigma (W_{ox}x_t + W_{oh}h_{t-1} + b_o)
\end{equation}
\begin{equation}
C_t = g_t\times i_t + C_{t-1} \times f_t
\end{equation}
\begin{equation}
h_t =  tanh (C_t) \times O_t
\end{equation}
The ensemble model combines the prediction results of the Linear and LSTM models through the weight coefficient. The weight coefficient is calculated by the minimum sum of squared errors, and the model with higher prediction accuracy is given greater weight. The weights are calculated as shown in Eq.(21).
\begin{equation}
W_i =  E_i/ \sum_{i=1}^{n}E_i
\end{equation}
Where $W_i$ is the weight of $Model_i$. $E_i$ is sum of squared errors of $Model_i$. $n$ is the type of prediction model, here is Linear and LSTM.  
\subsection{Frequent Access Timestamp Estimation}
After the query load arrival rate is predicted, we know the time series and related field which will be accessed next, because the query in the temporal data is more concerned about the data in a period time, i.e., the $record$ records a period time. Then the timestamp range of these data will be accessed frequently? The core lies in how to count the frequency of the accessed timestamp ranges in the query log. Here, we cut the period time into time points, so the problem is transformed into the problem to count the frequency of accessing the timestamp in the query log, which is a classic problem as follows~\cite{MG}.

\textbf{Problem Definition} \quad Given timestamp data streams $\sigma = \left\{a_1,a_2,...,a_m\right\}$, where $m$ is the size of the data stream, $a_i \in \left\{1,2,...,n\right\}$, The number of occurrences of each element is $f = (f_1,f_2,...,f_n)$, where $f_i$ is the number of occurrences of the $i$ element. Easily derived $m = \sum_{i=1}^{n}f_i$, given the parameter $k$, if we find all the elements with more occurrences than $m/k$, it is the output set $\left\{j|f_j>m/k\right\}$.

\indent We consider that if we maintain a counter for each element, we need to have $n$ counters that can be computed in $O(n)$ time. However, we do not have enough memory to store the whole data stream. Therefore, we approximate the problem based on the Misra-Gries algorithm with an error rate of $\varepsilon$.

Frequent Access Timestamp Mining Algorithm is shown in Algorithm 4. First, we create an array $T$ of size $k$. For $a_i$ arriving in the data stream in sequence, process it as follows: if $a_i$ is in the $T$, plus one to its counter(Lines 2-4); if $a_i$ is not in the $T$ and the number of elements in the $T$ is less than $k-1$, add the $a_i$ to the $T$ and assign a counter of 1 to it; in other cases, subtract 1 to the counters of all elements in $T$(Lines 5-9). If the counter value of an element in $T$ is 0, the element is removed from the $T$(Lines 10-11). When the scan of the data stream is completed, the $k$ elements saved in $T$ are the frequent items in the data stream. The time complexity is $O(m\times n)$. The space complexity is $O(\varepsilon^{-1}logmn)$. For any $query_i$, $\widehat{f_i}$ that is satisfied $f_i-\varepsilon n\leq \widehat{f_i} \leq f_i$ is returned. Therefore, the approximation ratio bound of Algorithm 4 is $\varepsilon$. After obtaining the frequent items in the timestamped data stream, combined with the query arrival rate prediction result, we can preheat dataset within the preheat-size constraints.
\begin{algorithm}[t]
	\SetAlgoLined 
	\caption{Frequent Access Timestamp Mining Algorithm}
	\KwIn{An array $T$ consists of ($x_i$,$C_{x_i}$), $k$ is~ size~of~$T$~, $\sigma$ = $\left\{a_1,a_2,...,a_m\right\}$}
	\KwOut{array $T$}
	$T$ $\gets$ $\emptyset$ \; 
	\For{i=0;~i~\textless ~m~;~i++}{
		
		\If{$a_i$ $\in$ $T$}{
			$(a_i,C_{a_i})$ $\gets$ $(a_i,C_{a_i}+1)$\;
		}
		\ElseIf{$a_i$ $\notin$ $T$}{
			\If{T.size\,\textless\,k\,-\,1}{
				$T$ $\gets$ $(a_i,1)$
			}
			\For{i=0;~i~\textless ~T.size~;~i++}{
				$(a_i,C_{a_i})$ $\gets$ $(a_i,C_{a_i}-1)$\;
			}
		}
		\If{$C_{a_i}=0$}{
			drop $(a_i,C_{a_i})$\;
		}
	}
	return	
\end{algorithm}
\section{Evaluation}
\subsection{Experimental Settings}
To evaluate the performance of TS-Cabinet, we implemented a hot and cold data hierarchical storage scheduler at the cloud-edge-end time-series database, which consists of an end device (windows system Intel(R) Core(TM) i7-8565U CPU @ 1.80GHz 1.99 GHz), an edge server (Ubuntu 18.04.6 LTS AMD(R) RyZen 5 3600x 6-core processor x 12), and a cloud server consisting of InfluxDB 1.1.1~\cite{influxdb} on them. We performed experiments using the pollution public dataset~\cite{pol}, which is a collection of air pollution measurements generated based on the air quality metric (449 observations) measuring data from August 2014 - October 2014. It has 200,000 data points, and the dataset size is 313 MB.

\subsection{Validity of the Temperature Model}

In order to verify the validity of the data temperature model of TS-Cabinet, we compared it with TITLE, the Newton's law of cooling combined with the temperature where the $T_r$ is a constant in \cite{xie}. The experimental results are shown in Figure 3 show that: (1) When $A$ and $B$ are accessed only once in the same time, such as period time 1-3, the later the access, the higher the temperature is, so $T_{B_3} \textgreater T_{A_3}$. (2) When $A$ and $B$ are accessed twice in the same time, such as the 1-4 period time, the smaller the access interval, the higher the temperature of the data. $T_{B_4} \textgreater T_{A_4}$. (3) When the access number is high, the data temperature is not necessarily high. When there are many visits, the access time is too early, it will be lower than the data temperature when there are few visits and the access time is late. Such as period time 4-11, $T_{B_{11}} \textless T_{A_{11}}$.

The above three observations are also satisfied TITLE. However, our model portrays the relationship between the warming magnitude and the access interval. The smaller the interval between two accesses is, the faster the warming magnitude is. It has the advantage of being able to reduce data transfer without migration to the cloud side and the edge side for cold data whose workload predicts that the interval between two accesses is too long.

\begin{figure}[htbp]
	\centering
	\setlength{\abovecaptionskip}{0cm} 
	\setlength{\belowcaptionskip}{0cm}
	\includegraphics[width=0.6\textwidth]{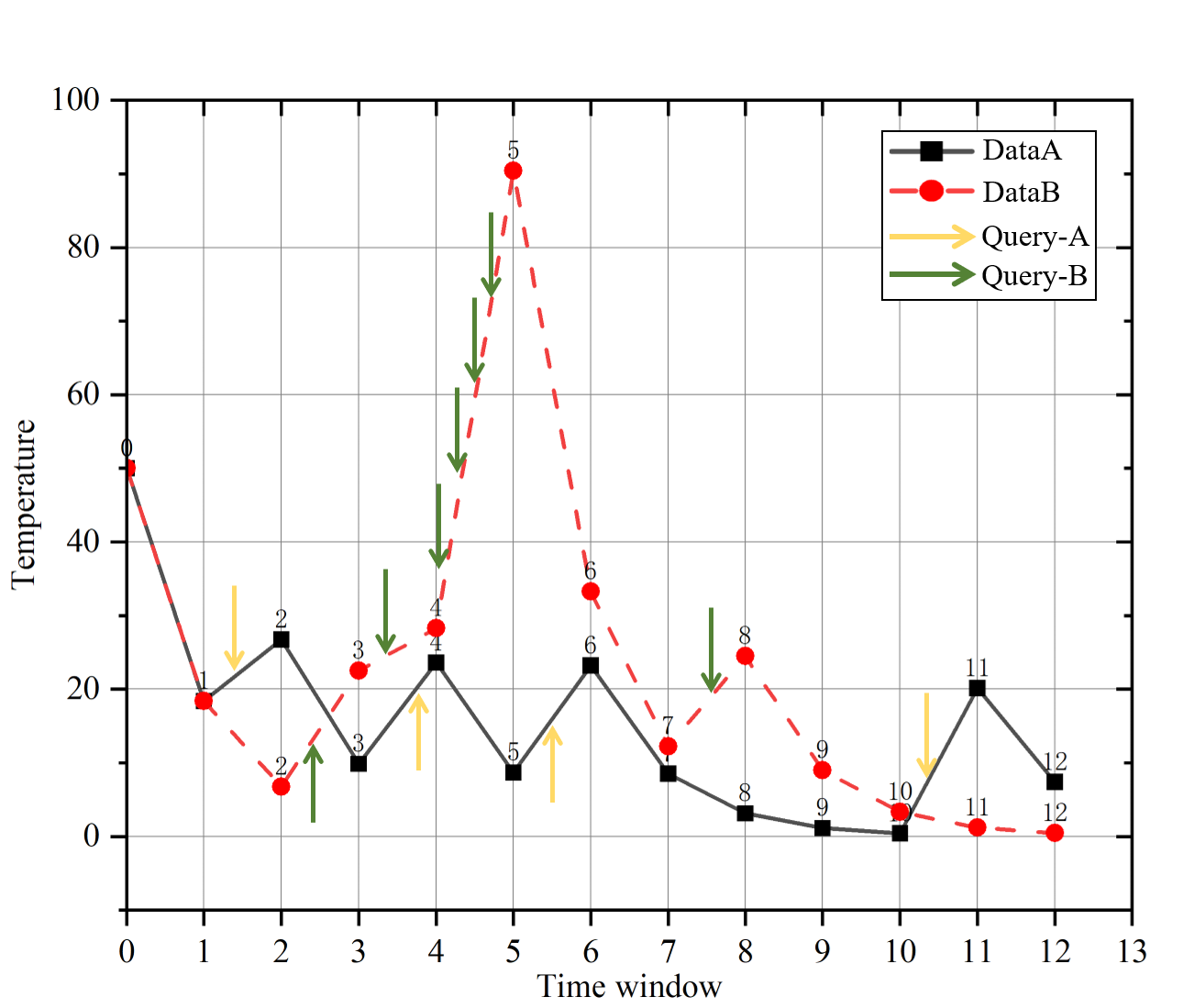}
	\caption{TS-Cabinet data temperature.} \label{fig1}
\end{figure}


\subsection{Workload Forecasting Accuracy Evaluation}
To verify the accuracy of the TS-Cabinet workload prediction model, we use Mean Absolute Error (MSE), Mean Square Error (MAE), and Root Mean Square Error (RMSE) as metrics. Moreover, TS-Cabinet is compared with the four most widely used prediction models, namely Linear, ARIMA, Holt winter's and LSTM. We use 5 mins as the prediction interval and record the performance of each model in eight prediction horizons. In addition, how far into the future a model can predict is called its prediction horizons. The experimental results are shown in Figure 4-7. Through Figures 4, 5, and 6, it can be seen that TS-Cabinet outperforms these four models in accuracy within a prediction horizon of 12 hours. After 12 hours, the accuracy of TS-Cabinet is higher than Holt winter's. In terms of training time, since TS-Cabinet is an ensemble model, the training time is relatively long, but it is almost equal to LSTM model.

\begin{figure}[htbp]
	\centering
	\setlength{\abovecaptionskip}{0cm} 
	\setlength{\belowcaptionskip}{0cm}
	\includegraphics[width=0.9\textwidth]{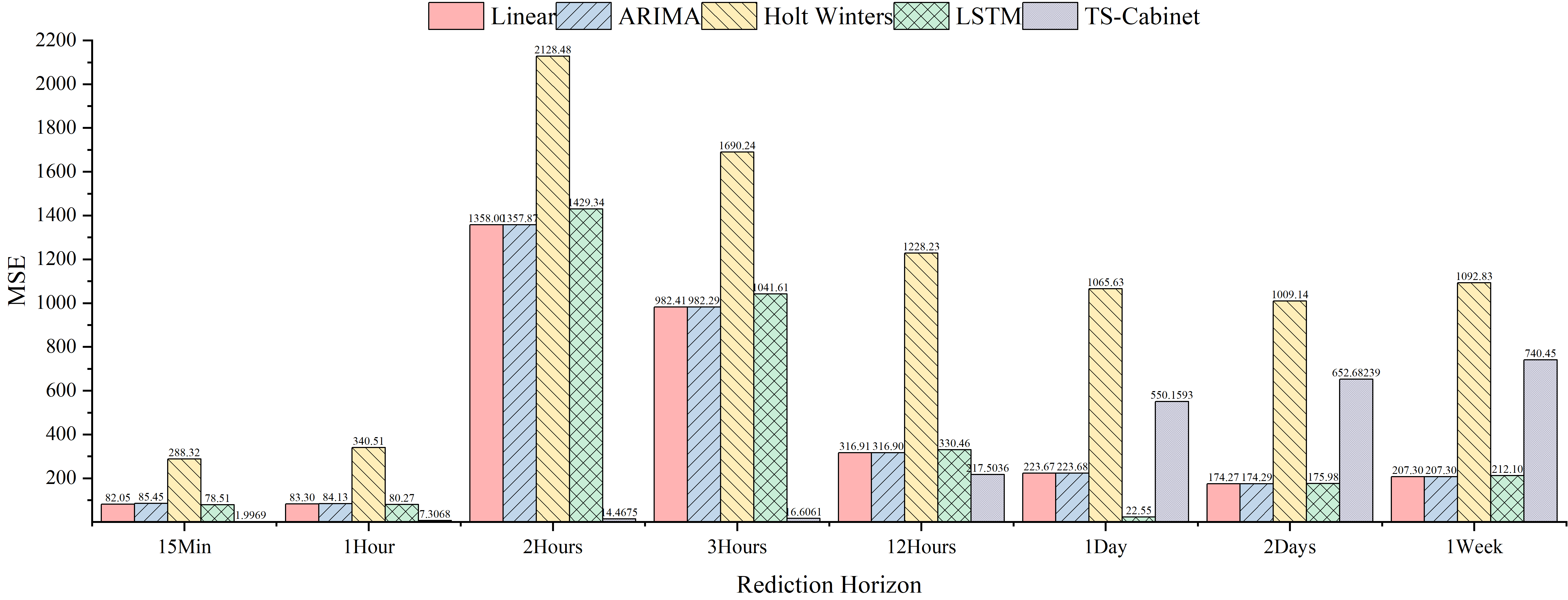}
	\caption{The MSE of the different forecasting models.} \label{fig1}
\end{figure}
\begin{figure}[htbp]
	\centering
	\setlength{\abovecaptionskip}{0cm} 
	\setlength{\belowcaptionskip}{0cm}
	\includegraphics[width=0.9\textwidth]{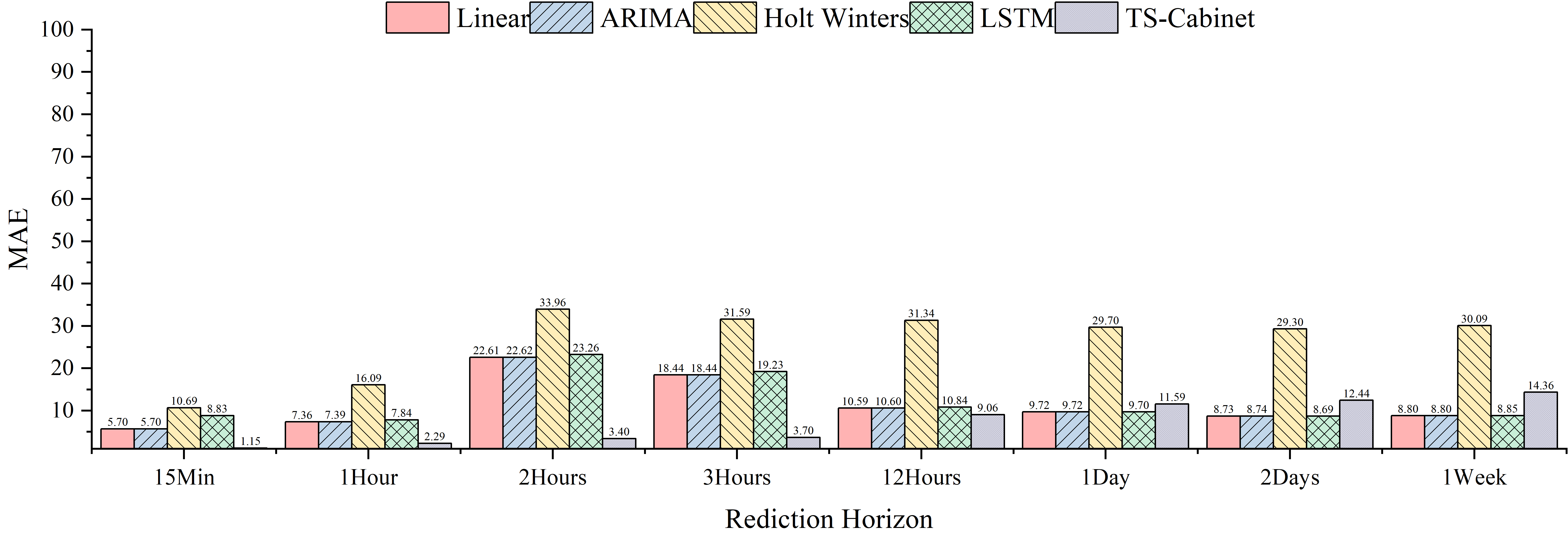}
	\caption{The MAE of the different forecasting models.} \label{fig1}
\end{figure}

\begin{figure}[htbp]
	\centering
	\setlength{\abovecaptionskip}{0cm} 
	\setlength{\belowcaptionskip}{0cm}
	\includegraphics[width=0.9\textwidth]{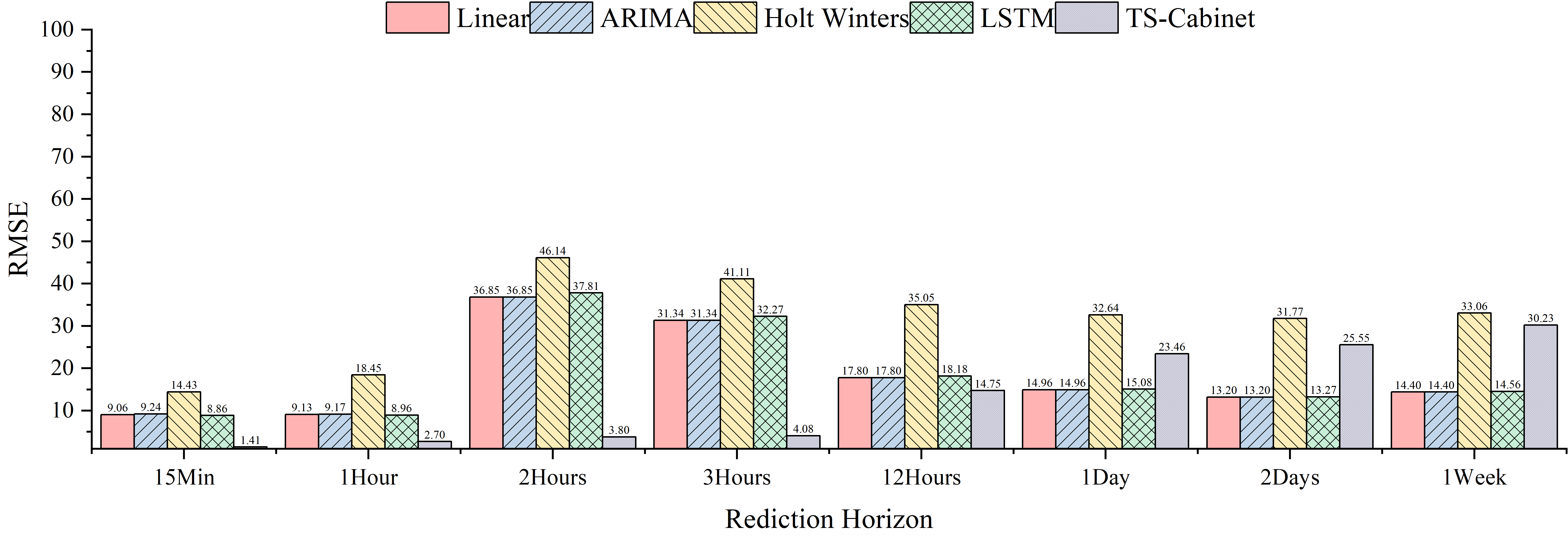}
	\caption{The RMSE of the different forecasting models.} \label{fig1}
\end{figure}

\begin{figure}[htbp]
	\centering
	\setlength{\abovecaptionskip}{0cm} 
	\setlength{\belowcaptionskip}{0cm}
	\includegraphics[width=0.6\textwidth]{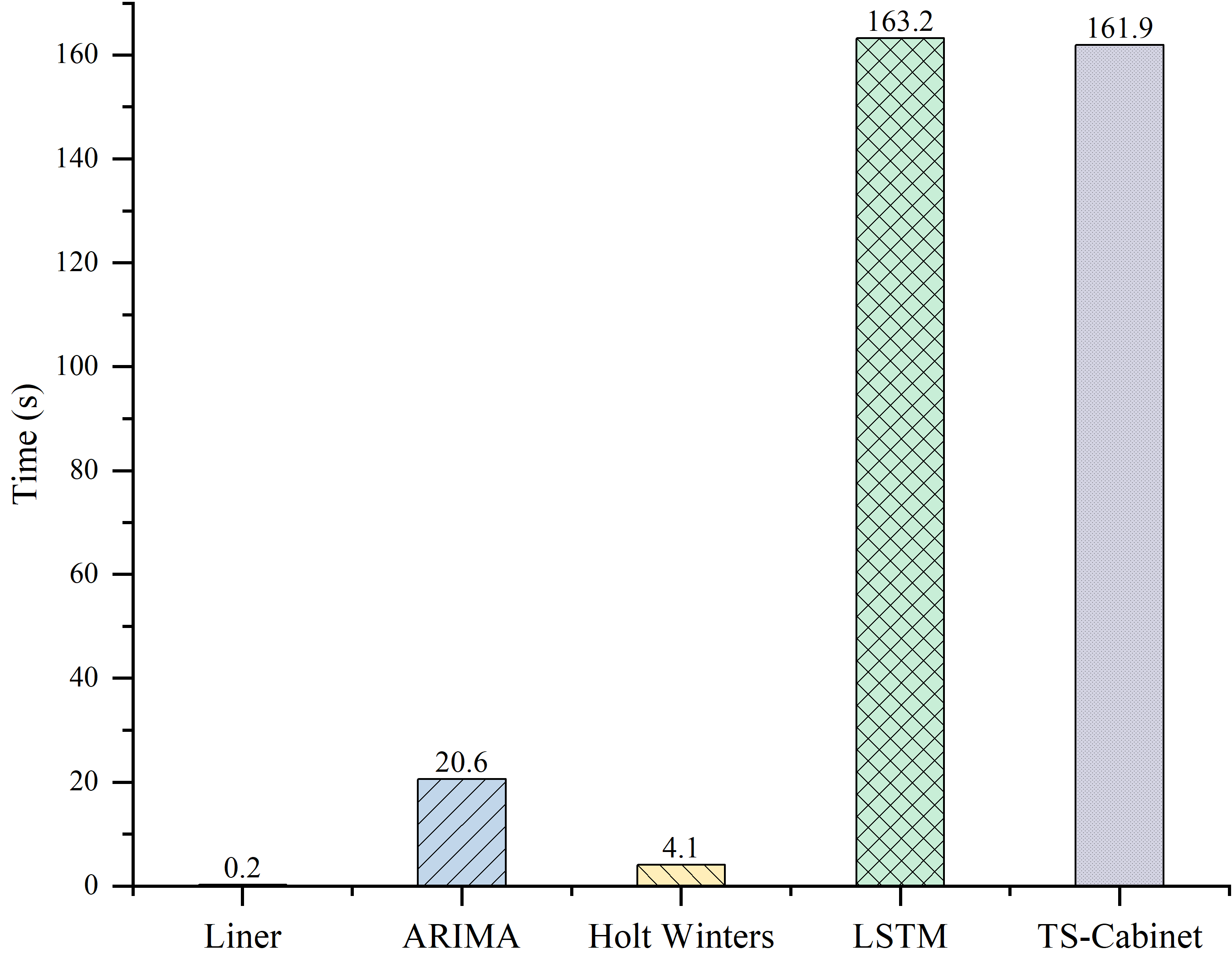}
	\caption{The training time of the different forecasting models.} \label{fig1}
\end{figure}


\subsection{Access Hit Rate Evaluation}
To verify TS-Cabinet how to improve the data access hit rate for cloud-side and edge-side, we compared it to TITLE\cite{xie}. Figure 8 shows the performance of the TITLE and TS-Cabinet with changeable value of hot database size. It is clear to find that with the increment of cache size, the hit ratio of all three models improved gradually. From start to end, TS-Cabinet with workload forecasting shows the highest cache hit ratio. Such a result confirms that the cold-hot data recognition and scheduling can benefit from workload forecasting, as workload forecaster can assist us to identify and preheat data points which seem to be visited frequently shortly. The rows of data with higher temperature are more likely to be loaded and kept at cloud-side and edge-side. It is worth to mention that TS-Cabinet without workload forecasting module achieves a much better performance than the TITLE. We owe it to the successful usage of our data temperature model, which means that leveraging the law of thermal radiation to model the access heating of data is more suitable for industrial IoT scenarios.

Through the above experiments, TS-Cabinet can achieve about 94\% hit rate for data access on the cloud side and edge side, which is 12\% better than
the existing methods. It shows that our proposed temperature model and workload prediction model can reduce the storage overhead of cloud-edge-end time-series database and improve the efficiency of collaborative queries.

\begin{figure}[htbp]
	\centering
	\setlength{\abovecaptionskip}{0cm} 
	\setlength{\belowcaptionskip}{0cm}
	\includegraphics[width=0.6\textwidth]{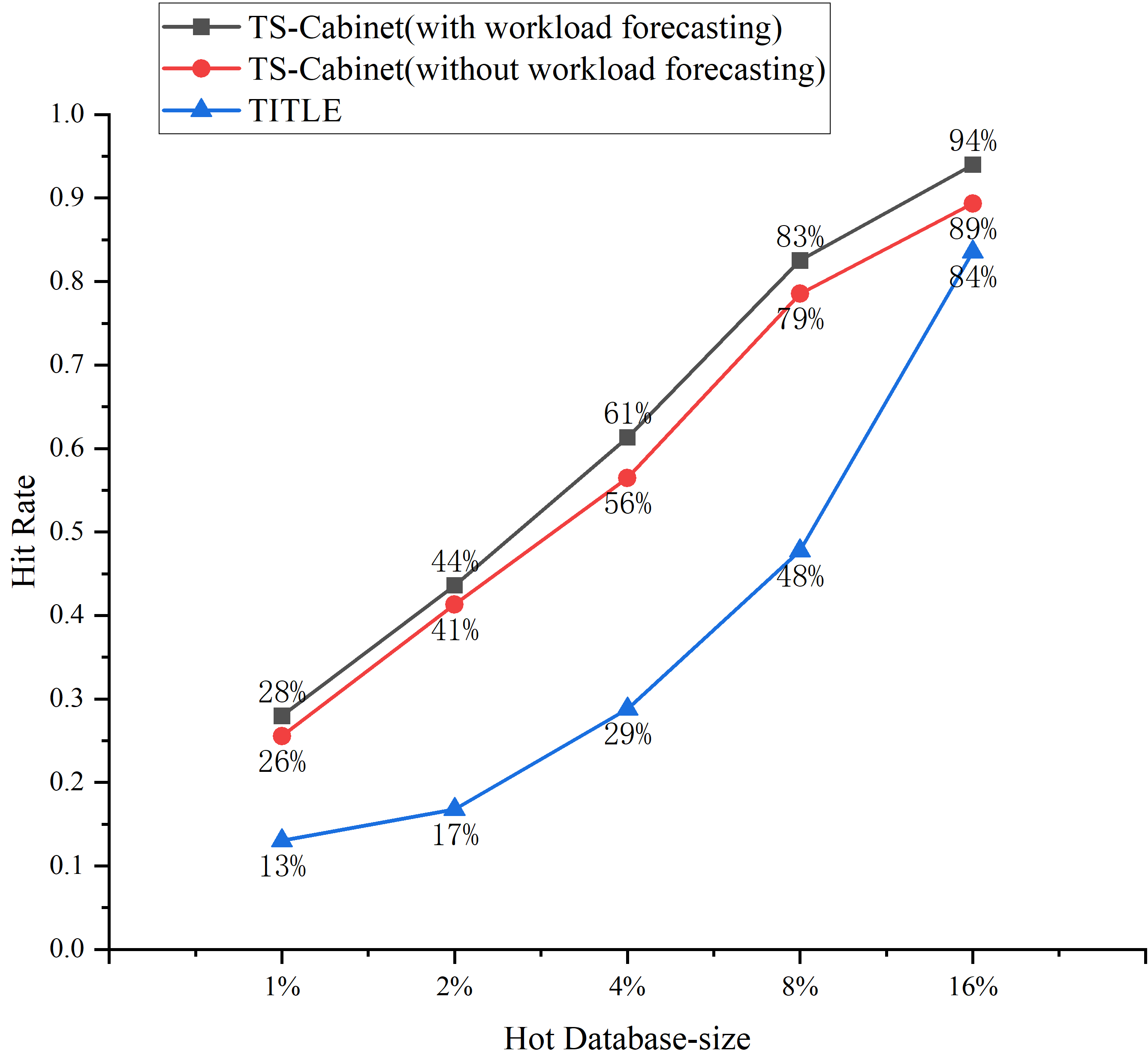}
	\caption{Access hit rate comparison.} \label{fig1}
\end{figure}

\section{Conclusion}

We propose a hierarchical storage strategy for cloud-edge-end time-series database. TS-Cabinet calculates the current temperature of data and uses workload forecasting and frequent access timestamp methods to calculate future temperature changes. It migrates data based on them. The experimental results show that TS-Cabinet can effectively help data in the cloud-edge-end time-series database to select optimal storage location. TS-Cabinet not only reduces storage overhead, but also improves the query efficiency. We plan to consider designing adaptive workload forecasting models to help query templates with different query arrival rate patterns to match the best workload forecasting models for future work.\\

\noindent\textbf{Acknowledgements}\\
\indent This paper was supported by NSFC grant  (62232005, U1866602).

%
%

%
%
%
%

\label{}





\begin{thebibliography}{31}
	
	\bibitem{chong}
	Pathak A, Gurajada A, Khadilkar P. Life cycle of transactional data in in-memory databases[C]//2018 IEEE 34th International Conference on Data Engineering Workshops (ICDEW). IEEE, 2018: 122-133.
	
	\bibitem{hash}
	Hsieh J W, Kuo T W, Chang L P. Efficient identification of hot data for flash memory storage systems[J]. ACM Transactions on Storage, 2006, 2(1): 22-40. 
	
	\bibitem{bloom}
	HPark D, Du D H C. Hot data identification for flash-based storage systems using multiple bloom filters[C]//2011 IEEE 27th Symposium on Mass Storage Systems and Technologies (MSST). IEEE, 2011: 1-11. 
	
	\bibitem{Siberia}
	Levandoski J J, Larson P Å, Stoica R. Identifying hot and cold data in main-memory databases[C]//2013 IEEE 29th International Conference on Data Engineering (ICDE). IEEE, 2013: 26-37.
	
	
	\bibitem{xie}
	Yulin Xie.TITLE:Research on Recognition Mechanismof Hot and Cold Data Based on Data Temperature[D]. Zhejiang University, 2019.
	
	
	\bibitem{xu}
	Jiaxin Xu.T he Research and Implementation of Storage Mechanism for Hot and Cold Data[D]. University of Electronic Science and Technology, 2021.
	
	\bibitem{Gorilla}
	Pelkonen T, Franklin S, Teller J, et al. Gorilla: A fast, scalable, in-memory time series database[J]. Proceedings of the VLDB Endowment, 2015, 8(12): 1816-1827.
	
	
	\bibitem{PB}
	Philippe Bonnet, Johannes Gehrke, Praveen Seshadri: Towards Sensor Database Systems. Mobile Data Management 2001: 3-14.
	
	\bibitem{OD}
	Ousmane D, Joel J. P. C. R, Mbaye S.: Real-time data management on wireless sensor networks: A survey. J. Netw. Comput. Appl. 35(3): 1013-1021 (2012).
	
	\bibitem{X}
	Akimitsu Kanzaki, Takahiro Hara, Yoshimasa Ishi, Tomoki Yoshihisa, Yuuichi Teranishi, Shinji Shimojo: X-Sensor: Wireless Sensor Network Testbed Integrating Multiple Networks. Wireless SensorNetwork Technologies for the Information Explosion Era 2010: 249-271.
	
	\bibitem{Sensor}
	Fei Hu, Xiaojun Cao.: Sensor Data Management. SCRC Press 2010: 237–257.
	
	\bibitem{AS}
	Srikrishna Prasad, Avinash Shankaranarayanan.: Smart meter data analytics using OpenTSDB and Hadoop. ISGT Asia 2013: 1-6.
	
	\bibitem{IOT}
	https://iotdb.apache.org/
	
	\bibitem{tao}
	https://taosdata.com/cn/
	
	\bibitem{CIKM}
	Holze M, Ritter N. Towards workload shift detection and prediction for autonomic databases[C]//Proceedings of the ACM Ph.D.workshop in CIKM. 2007: 109-116.
	
	\bibitem{QB}
	Ma L, Aken D V, Hefny A, et al. Query-based Workload Forecasting for Self-Driving Database Management Systems[C]//the International Conference. 2018.
	
	
	\bibitem{RNN}
	Liu C, Mao W, Gao Y, et al. Adaptive recollected RNN for workload forecasting in database-as-a-service[C]//International Conference on Service-Oriented Computing. Springer, Cham, 2020: 431-438.
	
	\bibitem{TEALED}
	Huang X, Cheng Y, Gao X, et al. TEALED: A Multi-Step Workload Forecasting Approach Using Time-Sensitive EMD and Auto LSTM Encoder-Decoder[C]//International Conference on Database Systems for Advanced Applications. Springer, Cham, 2022: 706-713.
	
	
	
	\bibitem{R1}
	T. P. P. Council. Tpc benchmark h (decision support). Standard Specification, Revision, 1(0), 1999.
	
	\bibitem{R2}
	M. Poess and C. Floyd. New tpc benchmarks for decision support and web commerce. ACM Sigmod Record, 29(4):64–71, 2000.
	
	\bibitem{YCSB}
	2010. YCSB-TS. http://tsdbbench.github.io/YCSB-TS/.
	
	\bibitem{TSBS}
	timescale. Time series benchmark suite. https://github.com/timescale/tsbs, 2020.
	
	\bibitem{TSB}
	Y. Hao et al., "TS-Benchmark: A Benchmark for Time Series Databases," 2021 IEEE 37th International Conference on Data Engineering (ICDE), 2021, pp. 588-599, doi: 10.1109/ICDE51399.2021.00057.
	
	\bibitem{DTW}
	Rakthanmanon T,  Campana B,  Mueen A, et al. Searching and Mining Trillions of Time Series Subsequences under Dynamic Time Warping[J]. SIGKDD explorations, 2012(CD/ROM):2012.
	
	\bibitem{KD}
	J. L. Bentley. Multidimensional binary search trees used for associative searching. Communications of the ACM, 18(9):509–517, 1975.
	
	\bibitem{DW}
	D. W. Opitz and R. Maclin. Popular ensemble methods: An empirical study. 1999.
	\bibitem{RP}
	R. Polikar. Ensemble based systems in decision making. IEEE Circuits and systems magazine, 6(3):21–45.
	\bibitem{ZH}
	Z.-H. Zhou. Ensemble methods: foundations and algorithms. CRC press, 2012.
	
	\bibitem{MG}
	Misra J, Gries D (1982) Finding repeated elements. Science of Computer Programming 2:143–152
	
	\bibitem{influxdb}
	2021. https://docs.influxdata.com/influxdb/v1.7/concepts/insights\_tradeoffs/.
	
	\bibitem{pol}
	2012. Air pollution index. https://en.wikipedia.org/wiki/Air\_Pollution\_Index.
	
	
	
\end{thebibliography}


\end{document}